\documentclass{article}%
\usepackage{amssymb}
\usepackage{amsmath}
\usepackage{amsfonts}
\usepackage{graphicx}%
\providecommand{\U}[1]{\protect\rule{.1in}{.1in}}
\newtheorem{theorem}{Theorem}[section]

\newtheorem{lemma}[theorem]{Lemma}

\newtheorem{problem}[theorem]{Problem}

\newtheorem{remark}[theorem]{Remark}

\newenvironment{proof}[1][Proof]{\noindent\textbf{#1.} }{\ \rule{0.5em}{0.5em}}
\begin{document}

\author{Vadim E. Levit\\Ariel University Center of Samaria, Israel\\levitv@ariel.ac.il
\and Eugen Mandrescu\\Holon Institute of Technology, Israel\\eugen\_m@hit.ac.il}
\title{On the Intersection of All Critical Sets of a Unicyclic Graph}
\date{}
\maketitle

\begin{abstract}
A set $S\subseteq V$ is \textit{independent} in a graph $G=\left(  V,E\right)
$ if no two vertices from $S$ are adjacent. The \textit{independence number}
$\alpha(G)$ is the cardinality of a maximum independent set, while $\mu(G)$ is
the size of a maximum matching in $G$. If $\alpha(G)+\mu(G)=\left\vert
V\right\vert $, then $G$ is called a \textit{K\"{o}nig-Egerv\'{a}ry graph}.
The number $d_{c}\left(  G\right)  =\max\{\left\vert A\right\vert -\left\vert
N\left(  A\right)  \right\vert :A\subseteq V\}$ is called the \textit{critical
difference} of $G$ \cite{Zhang}, where $N\left(  A\right)  =\left\{  v:v\in
V,N\left(  v\right)  \cap A\neq\emptyset\right\}  $. By \textrm{core}$\left(
G\right)  $ (\textrm{corona}$(G)$) we denote the intersection (union,
respectively) of all maximum independent sets, while by $\ker\left(  G\right)
$ we mean the intersection of all critical independent sets. A connected graph
having only one cycle is called \textit{unicyclic}.

It is known that the relation $\ker\left(  G\right)  \subseteq$ \textrm{core}%
$\left(  G\right)  $ holds for every graph $G$ \cite{Levman2011a}, while the
equality is true for bipartite graphs \cite{Levman2011b}. For
K\"{o}nig-Egerv\'{a}ry unicyclic graphs, the difference $\left\vert
\mathrm{core}(G)\right\vert -\left\vert \ker\left(  G\right)  \right\vert $
may equal any non-negative integer. 

In this paper we prove that if $G$ is a non-K\"{o}nig-Egerv\'{a}ry unicyclic
graph, then: \emph{(i) }$\ker\left(  G\right)  =$ \textrm{core}$\left(
G\right)  $ and \emph{(ii) }$\left\vert \mathrm{corona}(G)\right\vert
+\left\vert \mathrm{core}(G)\right\vert =2\alpha\left(  G\right)  +1$. Pay
attention that $\left\vert \mathrm{corona}(G)\right\vert +\left\vert
\mathrm{core}(G)\right\vert =2\alpha\left(  G\right)  $ holds for every
K\"{o}nig-Egerv\'{a}ry graph \cite{Levman2011b}.


\end{abstract}

\section{Introduction}

Throughout this paper $G=(V,E)$ is a simple (i.e., a finite, undirected,
loopless and without multiple edges) graph with vertex set $V=V(G)$ and edge
set $E=E(G)$. If $X\subset V$, then $G[X]$ is the subgraph of $G$ spanned by
$X$. By $G-W$ we mean the subgraph $G[V-W]$, if $W\subset V(G)$. For $F\subset
E(G)$, by $G-F$ we denote the partial subgraph of $G$ obtained by deleting the
edges of $F$, and we use $G-e$, if $W$ $=\{e\}$. If $A,B$ $\subset V$ and
$A\cap B=\emptyset$, then $(A,B)$ stands for the set $\{e=ab:a\in A,b\in
B,e\in E\}$. The neighborhood of a vertex $v\in V$ is the set $N(v)=\{w:w\in
V$ \textit{and} $vw\in E\}$, and $N(A)=\cup\{N(v):v\in A\}$, $N[A]=A\cup N(A)$
for $A\subset V$. By $C_{n},K_{n}$ we mean the chordless cycle on $n\geq$ $4$
vertices, and respectively the complete graph on $n\geq1$ vertices.

A set $S$ of vertices is \textit{independent} if no two vertices from $S$ are
adjacent, and an independent set of maximum size will be referred to as a
\textit{maximum independent set}. The \textit{independence number }of $G$,
denoted by $\alpha(G)$, is the cardinality of a maximum independent
set\textit{\ }of $G$.

Let \textrm{core}$(G)=\cap\{S:S\in\Omega(G)\}$ \cite{levm3}, and
\textrm{corona}$(G)=\cup\{S:S\in\Omega(G)\}$ \cite{BorosGolLev}, where
$\Omega(G)=\{S:S$ \textit{is a maximum independent set of} $G\}$.

\begin{theorem}
\cite{BorosGolLev}\label{th11} For every $S\in\Omega\left(  G\right)  $, there
is a matching from $S-$\textrm{core}$(G)$ into \textrm{corona}$(G)-S$.
\end{theorem}

An edge $e\in E(G)$ is $\alpha$-\textit{critical} whenever $\alpha
(G-e)>\alpha(G)$. Notice that $\alpha(G)\leq\alpha(G-e)\leq\alpha(G)+1$ holds
for each edge $e$.

The number $d(X)=\left\vert X\right\vert -\left\vert N(X)\right\vert
,X\subseteq V(G)$, is called the \textit{difference} of the set $X$, while
$d_{c}(G)=\max\{d(X):X\subseteq V\}$ is called the \textit{critical
difference} of $G$. A set $U\subseteq V(G)$ is \textit{critical} if
$d(U)=d_{c}(G)$ \cite{Zhang}. The number $id_{c}(G)=\max\{d(I):I\in
\mathrm{Ind}(G)\}$ is called the \textit{critical independence difference} of
$G$. If $A\subseteq V(G)$ is independent and $d(A)=id_{c}(G)$, then $A$ is
called a \textit{critical independent set }\cite{Zhang}. Clearly,
$d_{c}(G)\geq id_{c}(G)$ is true for every graph $G$.

\begin{theorem}
\cite{Zhang} The equality $d_{c}(G)$ $=id_{c}(G)$ holds for every graph $G$.
\end{theorem}

For a graph $G$, let denote $\mathrm{\ker}(G)=\cap\left\{  S:S\text{
\textit{is a critical independent set}}\right\}  $.

\begin{theorem}
\label{th2}If $G$ is a graph, then

\emph{(i)} \cite{Levman2011a} $\mathrm{\ker}\left(  G\right)  $ is a critical
independent set and $\mathrm{\ker}\left(  G\right)  \subseteq\mathrm{core}(G)$;

\emph{(ii)} \cite{Levman2011b} $\mathrm{\ker}\left(  G\right)  =\mathrm{core}%
(G)$, whenever $G$ is bipartite.
\end{theorem}

A matching (i.e., a set of non-incident edges of $G$) of maximum cardinality
$\mu(G)$ is a \textit{maximum matching}, and a \textit{perfect matching} is
one covering all vertices of $G$. An edge $e\in E(G)$ is $\mu$%
-\textit{critical }provided $\mu(G-e)<\mu(G)$.

It is well-known that $\lfloor n/2\rfloor+1\leq\alpha(G)+\mu(G)\leq n$ hold
for any graph $G$ with $n$ vertices. If $\alpha(G)+\mu(G)=n$, then $G$ is
called a \textit{K\"{o}nig-Egerv\'{a}ry graph }\cite{dem}, \cite{ster}.
Several properties of K\"{o}nig-Egerv\'{a}ry graphs are presented in
\cite{levm2}, \cite{LevMan2003}, \cite{LevMan5}.

According to a celebrated result of K\"{o}nig, \cite{koen}, and Egerv\'{a}ry,
\cite{eger}, any bipartite graph is a K\"{o}nig-Egerv\'{a}ry\emph{ }graph.
This class includes also non-bipartite graphs (see, for instance, the graph
$G$ in Figure \ref{fig112}).

\begin{figure}[h]
\setlength{\unitlength}{1cm}\begin{picture}(5,1.8)\thicklines
\multiput(4,0.5)(1,0){5}{\circle*{0.29}}
\multiput(5,1.5)(2,0){2}{\circle*{0.29}}
\put(4,0.5){\line(1,0){4}}
\put(5,0.5){\line(0,1){1}}
\put(7,1.5){\line(1,-1){1}}
\put(7,0.5){\line(0,1){1}}
\put(4,0.1){\makebox(0,0){$a$}}
\put(4.7,1.5){\makebox(0,0){$b$}}
\put(6,0.1){\makebox(0,0){$c$}}
\put(5,0.1){\makebox(0,0){$u$}}
\put(7,0.1){\makebox(0,0){$v$}}
\put(6.7,1.5){\makebox(0,0){$x$}}
\put(8,0.1){\makebox(0,0){$y$}}
\put(3.2,1){\makebox(0,0){$G$}}
\end{picture}
\caption{A K\"{o}nig-Egerv\'{a}ry graph with $\alpha(G)=\left\vert \left\{
a,b,c,x\right\}  \right\vert $ and $\mu(G)=\left\vert \left\{
au,cv,xy\right\}  \right\vert $.}%
\label{fig112}%
\end{figure}
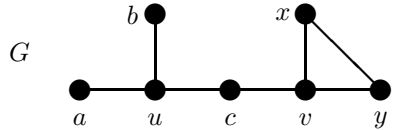

\begin{theorem}
\label{th1}\cite{LevMan2003} If $G$ is a K\"{o}nig-Egerv\'{a}ry graph, then
every maximum matching matches $N($\textrm{core}$(G))$ into \textrm{core}$(G)$.
\end{theorem}

The graph $G$ is called \textit{unicyclic} if it is connected and has a unique
cycle, which we denote by $C=\left(  V(C),E\left(  C\right)  \right)  $. Let%
\[
N_{1}(C)=\{v:v\in V\left(  G\right)  -V(C),N(v)\cap V(C)\neq\emptyset\},
\]
and $T_{x}=(V_{x},E_{x})$ be the tree of $G-xy$ containing $x$, where $x\in
N_{1}(C),y\in V(C)$.\begin{figure}[h]
\setlength{\unitlength}{1cm}\begin{picture}(5,1.8)\thicklines
\multiput(2,0.5)(1,0){6}{\circle*{0.29}}
\multiput(3,1.5)(1,0){4}{\circle*{0.29}}
\put(2,0.5){\line(1,0){5}}
\put(3,1.5){\line(1,-1){1}}
\put(4,0.5){\line(0,1){1}}
\put(5,1.5){\line(1,0){1}}
\put(6,1.5){\line(1,-1){1}}
\put(5,0.5){\line(0,1){1}}
\put(2,0.1){\makebox(0,0){$u$}}
\put(3,0.1){\makebox(0,0){$v$}}
\put(4,0.1){\makebox(0,0){$x$}}
\put(5,0.1){\makebox(0,0){$y$}}
\put(6,0.1){\makebox(0,0){$w$}}
\put(7,0.1){\makebox(0,0){$c$}}
\put(2.7,1.5){\makebox(0,0){$a$}}
\put(3.7,1.5){\makebox(0,0){$b$}}
\put(4.7,1.5){\makebox(0,0){$d$}}
\put(6.3,1.5){\makebox(0,0){$t$}}
\put(1.1,1){\makebox(0,0){$G$}}
\multiput(9,0.5)(1,0){3}{\circle*{0.29}}
\multiput(10,1.5)(1,0){2}{\circle*{0.29}}
\put(9,0.5){\line(1,0){2}}
\put(10,1.5){\line(1,-1){1}}
\put(11,0.5){\line(0,1){1}}
\put(9,0.1){\makebox(0,0){$u$}}
\put(10,0.1){\makebox(0,0){$v$}}
\put(11,0.1){\makebox(0,0){$x$}}
\put(9.7,1.5){\makebox(0,0){$a$}}
\put(10.7,1.5){\makebox(0,0){$b$}}
\put(8.25,1){\makebox(0,0){$T_{x}$}}
\end{picture}\caption{$G$ is a unicyclic non-K\"{o}nig-Egerv\'{a}ry graph with
$V(C)=\{y,d,t,c,w\}$.}%
\end{figure}
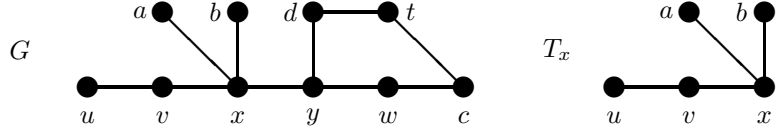

The following result shows that a unicyclic graph is either a
K\"{o}nig-Egerv\'{a}ry graph or each edge of its cycle is $\alpha$-critical.

\begin{lemma}
\cite{Levman2011d}\label{lem2} If $G$ is a unicyclic graph of order $n$, then

\emph{(i)} $n-1\leq\alpha(G)+\mu(G)\leq n$;

\emph{(ii)} $n-1=\alpha(G)+\mu(G)$ if and only if each edge of the unique
cycle is $\alpha$-critical.
\end{lemma}

\begin{theorem}
\cite{Levman2011d}\label{th12} Let $G$ be a unicyclic
non-K\"{o}nig-Egerv\'{a}ry graph. Then the following assertions are true:

\emph{(i) }each $W\in\Omega\left(  T_{x}\right)  $ can be enlarged to some
$S\in\Omega\left(  G\right)  $;

\emph{(ii)} $S\cap V\left(  T_{x}\right)  \in\Omega\left(  T_{x}\right)  $ for
every $S\in\Omega\left(  G\right)  $;

\emph{(iii)} $\mathrm{core}\left(  G\right)  =\cup\left\{  \mathrm{core}%
\left(  T_{x}\right)  :x\in N_{1}\left(  C\right)  \right\}  $.
\end{theorem}

Unicyclic graphs keep enjoying plenty of interest, as one can see , for
instance, in \cite{Belardo2010}, \cite{Du2010}, \cite{Huo2010},
\cite{LevMan2009a}, \cite{Li2010}, \cite{Wu2010}, \cite{Zhai2010}.

In this paper we analyze the relationship between several parameters of a
unicyclic graph $G$, namely, $\mathrm{core}(G)$, $\mathrm{corona}(G)$,
$\ker\left(  G\right)  $.

\section{Results}

\begin{lemma}
\label{lem1}If $G$ is a unicyclic non-K\"{o}nig-Egerv\'{a}ry graph, then

\emph{(i)} \textrm{core}$(G)\cap N\left[  V\left(  C\right)  \right]
=\emptyset$;

\emph{(ii)} there exists a matching from $N(\mathrm{core}(G))$ into
$\mathrm{core}(G)$.
\end{lemma}

\begin{proof}
\emph{(i)} Let\emph{ }$ab\in E\left(  C\right)  $. By Lemma \ref{lem2}%
\emph{(ii)}, the edge $ab$ is $\alpha$-critical. Hence there are $S_{a}%
,S_{b}\in\Omega\left(  G\right)  $, such that $a\in S_{a}$ and $b\in S_{b}$.
Since $a\notin S_{b}$, it follows that $a\notin$ \textrm{core}$(G)$, and
because $a\in S_{a}$, we infer that $N\left(  a\right)  \cap$ \textrm{core}%
$(G)=\emptyset$. Consequently, we obtain that \textrm{core}$(G)\cap N\left[
V\left(  C\right)  \right]  =\emptyset$.

\emph{(ii)} If $\mathrm{core}(G)=\emptyset$, then the conclusion is clear.

Assume that $\mathrm{core}(G)\neq\emptyset$. By Theorem \ref{th1}, in each
tree $T_{x}$ there is a matching $M_{x}$ from $N(\mathrm{core}(T_{x}))$ into
$\mathrm{core}(T_{x})$. By part \emph{(i)}, we have that $V(C)\cap N\left[
\mathrm{core}(G)\right]  =\emptyset$. Taking into account Theorem
\ref{th12}\emph{(ii)}, we see that the union of all these matchings $M_{x}$
gives a matching from $N(\mathrm{core}(G))$ into $\mathrm{core}(G)$.
\end{proof}

It is worth mentioning that the assertion in Lemma \ref{lem1}\emph{(ii)} is
true for every K\"{o}nig-Egerv\'{a}ry graph, by Theorem \ref{th1}. The graph
$G_{2}$ from Figure \ref{fig1122} shows that Lemma \ref{lem1}\emph{(i)} may
fail for unicyclic K\"{o}nig-Egerv\'{a}ry graphs. 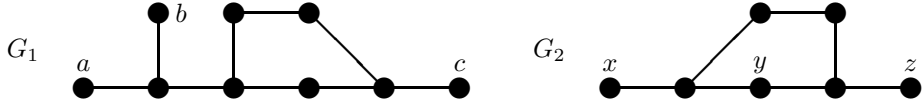
\begin{figure}[h]
\setlength{\unitlength}{1cm}\begin{picture}(5,1.8)\thicklines
\multiput(2,0)(1,0){6}{\circle*{0.29}}
\multiput(3,1)(1,0){3}{\circle*{0.29}}
\put(2,0){\line(1,0){5}}
\put(3,0){\line(0,1){1}}
\put(4,0){\line(0,1){1}}
\put(4,1){\line(1,0){1}}
\put(5,1){\line(1,-1){1}}
\put(2,0.3){\makebox(0,0){$a$}}
\put(3.3,1){\makebox(0,0){$b$}}
\put(7,0.3){\makebox(0,0){$c$}}
\put(1.2,0.5){\makebox(0,0){$G_{1}$}}
\multiput(9,0)(1,0){5}{\circle*{0.29}}
\multiput(11,1)(1,0){2}{\circle*{0.29}}
\put(9,0){\line(1,0){4}}
\put(10,0){\line(1,1){1}}
\put(11,1){\line(1,0){1}}
\put(12,0){\line(0,1){1}}
\put(9,0.3){\makebox(0,0){$x$}}
\put(11,0.3){\makebox(0,0){$y$}}
\put(13,0.3){\makebox(0,0){$z$}}
\put(8.2,0.5){\makebox(0,0){$G_{2}$}}
\end{picture}
\caption{K\"{o}nig-Egerv\'{a}ry graphs with \textrm{core}$(G_{1})=\left\{
a,b,c\right\}  $ and \textrm{core}$(G_{2})=\left\{  x,y,z\right\}  $.}%
\label{fig1122}%
\end{figure}

\begin{theorem}
\label{th4} If $G$ is a K\"{o}nig-Egerv\'{a}ry graph, then

\emph{(i)} \cite{LevMan2003} $N\left(  \mathrm{core}(G)\right)  =\cap
\{V(G)-S:S\in(G)\}$, i.e., $N\left(  \mathrm{core}(G\right)  )=V\left(
G\right)  -\mathrm{corona}(G)$;

\emph{(ii)} \cite{Levman2011b} $\left\vert \mathrm{corona}(G)\right\vert
+\left\vert \mathrm{core}(G)\right\vert =2\alpha\left(  G\right)  $.
\end{theorem}

\begin{figure}[h]
\setlength{\unitlength}{1cm}\begin{picture}(5,1.3)\thicklines
\multiput(1.5,0)(1,0){6}{\circle*{0.29}}
\multiput(2.5,1)(1,0){4}{\circle*{0.29}}
\put(1.5,0){\line(1,0){5}}
\put(2.5,0){\line(0,1){1}}
\put(3.5,0){\line(0,1){1}}
\put(3.5,1){\line(1,-1){1}}
\put(4.5,0){\line(0,1){1}}
\put(4.5,1){\line(1,0){1}}
\put(5.5,1){\line(1,-1){1}}
\put(1.6,0.3){\makebox(0,0){$a$}}
\put(2.8,1){\makebox(0,0){$b$}}
\put(4.7,0.3){\makebox(0,0){$c$}}
\put(0.6,0.5){\makebox(0,0){$G_1$}}
\multiput(8.5,0)(1,0){5}{\circle*{0.29}}
\multiput(9.5,1)(1,0){3}{\circle*{0.29}}
\put(8.5,0){\line(1,0){4}}
\put(11.5,1){\line(1,-1){1}}
\put(10.5,0){\line(0,1){1}}
\put(10.5,1){\line(1,0){1}}
\put(9.5,0){\line(0,1){1}}
\put(8.5,0.3){\makebox(0,0){$x$}}
\put(9.8,1){\makebox(0,0){$y$}}
\put(7.6,0.5){\makebox(0,0){$G_2$}}
\end{picture}\caption{Non-K\"{o}nig-Egerv\'{a}ry graphs with $\mathrm{core}%
(G_{1})=\left\{  a,b\right\}  $ and $\mathrm{core}(G_{2})=\{x,y\}.$}%
\label{fig5353}%
\end{figure}
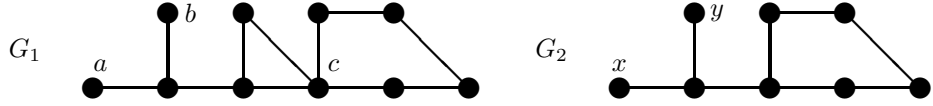

The graphs $G_{1}$, $G_{2}$ from Figure \ref{fig5353} satisfy $\mathrm{corona}%
(G_{1})\cup N\left(  \mathrm{core}(G_{1}\right)  )=V\left(  G_{1}\right)
-\left\{  c\right\}  $, while $\mathrm{corona}(G_{2})\cup N\left(
\mathrm{core}(G_{2}\right)  )=V\left(  G_{2}\right)  $.

\begin{theorem}
\label{th3}If $G$ is unicyclic non-K\"{o}nig-Egerv\'{a}ry graph, then
\begin{gather*}
\mathrm{corona}(G)\cup N\left(  \mathrm{core}(G\right)  )=V\left(  G\right)
,\\
\mathrm{corona}(G)=V\left(  C\right)  \cup\left(  \cup\left\{  \mathrm{corona}%
\left(  T_{x}\right)  :x\in N_{1}\left(  C\right)  \right\}  \right)  .
\end{gather*}

\end{theorem}

\begin{proof}
For the first equality, it is enough to show that $V\left(  G\right)
\subseteq\mathrm{corona}\left(  G\right)  \cup N\left(  \mathrm{core}\left(
G\right)  \right)  $.

Let $a\in V\left(  G\right)  $.

\textit{Case 1}. $a\in V\left(  C\right)  $. If $b\in N\left(  a\right)  \cap
V\left(  C\right)  $, then, by Lemma \ref{lem2}\emph{(ii)}, the edge $ab$ is
$\alpha$-critical. Hence $a\in\mathrm{corona}\left(  G\right)  $.

\textit{Case 2}. $a\in V\left(  G\right)  -V\left(  C\right)  $. Since
$V\left(  G\right)  =V\left(  C\right)  \cup\left(  \cup\left\{  V\left(
T_{x}\right)  :x\in N_{1}\left(  C\right)  \right\}  \right)  $, it follows
that there is some $y\in N_{1}\left(  C\right)  $, such that $a\in V\left(
T_{y}\right)  $. According to Theorem \ref{th4}\emph{(i)}, we know that
$V\left(  T_{y}\right)  =\mathrm{corona}(T_{y})\cup N\left(  \mathrm{core}%
(T_{y}\right)  )$.

By Theorem \ref{th12}\emph{(i)}, $\mathrm{corona}\left(  T_{x}\right)
\subseteq\mathrm{corona}\left(  G\right)  $ for every $x\in N_{1}\left(
C\right)  $. Therefore, either $a\in\mathrm{corona}\left(  T_{y}\right)
\subseteq\mathrm{corona}\left(  G\right)  $, or $a\in N\left(  \mathrm{core}%
\left(  T_{y}\right)  \right)  \subseteq N\left(  \mathrm{core}\left(
G\right)  \right)  $, because $\mathrm{core}\left(  T_{y}\right)
\subseteq\mathrm{core}(G)$ in accordance with Theorem \ref{th12}\emph{(iii)}.

Consequently, $a\in\mathrm{corona}\left(  G\right)  \cup N\left(
\mathrm{core}\left(  G\right)  \right)  $. In other words, we get that
\[
V\left(  G\right)  \subseteq\mathrm{corona}\left(  G\right)  \cup N\left(
\mathrm{core}\left(  G\right)  \right)  ,
\]
as required.

As for the second equality, let us notice that $V\left(  C\right)
\subseteq\mathrm{corona}\left(  G\right)  $, by \textit{Case 1}. If
$a\in\mathrm{corona}\left(  G\right)  -V\left(  C\right)  $, then by Theorem
\ref{th12}\emph{(ii)}, there are $S\in\Omega\left(  G\right)  $ and $b\in
N_{1}\left(  C\right)  $, such that $a\in S\cap V\left(  T_{b}\right)
\in\Omega\left(  T_{b}\right)  $. Hence $a\in\mathrm{corona}\left(
T_{x}\right)  $, and therefore,
\[
\mathrm{corona}(G)-V\left(  C\right)  \subseteq\cup\left\{  \mathrm{corona}%
\left(  T_{x}\right)  :x\in N_{1}\left(  C\right)  \right\}  .
\]
Theorem \ref{th12}\emph{(i)} assures that $\cup\left\{  \mathrm{corona}\left(
T_{x}\right)  :x\in N_{1}\left(  C\right)  \right\}  \subseteq\mathrm{corona}%
\left(  G\right)  $. In conclusion, $\mathrm{corona}(G)=V\left(  C\right)
\cup\left(  \cup\left\{  \mathrm{corona}\left(  T_{x}\right)  :x\in
N_{1}\left(  C\right)  \right\}  \right)  $.
\end{proof}

The graph $G_{2}$ from Figure \ref{fig5353} shows that the equality
$\left\vert \mathrm{corona}(G)\right\vert +\left\vert \mathrm{core}%
(G)\right\vert =2\alpha\left(  G\right)  $ is not true for unicyclic
non-K\"{o}nig-Egerv\'{a}ry graphs.

\begin{theorem}
If $G$ is a unicyclic graph, then
\[
2\alpha\left(  G\right)  \ \leq\left\vert \mathrm{corona}(G)\right\vert
+\left\vert \mathrm{core}(G)\right\vert \leq2\alpha\left(  G\right)  +1.
\]

Moreover, $G$ is a non-K\"{o}nig-Egerv\'{a}ry graph if and only if
\[
\left\vert \mathrm{corona}(G)\right\vert +\left\vert \mathrm{core}%
(G)\right\vert =2\alpha\left(  G\right)  +1.
\]

\end{theorem}

\begin{proof}
By Theorem \ref{th4}\emph{(ii)}, the equality $2\alpha\left(  G\right)
=\left\vert \mathrm{corona}(G)\right\vert +\left\vert \mathrm{core}%
(G)\right\vert $ holds for every unicyclic K\"{o}nig-Egerv\'{a}ry graph $G$.

Assume now that $G$ is not a K\"{o}nig-Egerv\'{a}ry graph.

Let $S\in\Omega\left(  G\right)  $. According to Theorem \ref{th3} and Lemma
\ref{lem2}\emph{(i)}, we infer that
\[
\left\vert S\right\vert +\left\vert \mathrm{corona}(G)-S\right\vert
+\left\vert N\left(  \mathrm{core}\left(  G\right)  \right)  \right\vert
=\left\vert V\left(  G\right)  \right\vert =\alpha\left(  G\right)
+\mu\left(  G\right)  +1,
\]
which implies $\left\vert \mathrm{corona}(G)-S\right\vert +\left\vert N\left(
\mathrm{core}\left(  G\right)  \right)  \right\vert =\mu\left(  G\right)  +1$.

By Theorem \ref{th11}, there is a matching $M_{1}$ from $S-\mathrm{core}%
\left(  G\right)  $ into $\mathrm{corona}(G)-S$, which implies $\left\vert
S-\mathrm{core}\left(  G\right)  \right\vert \leq\left\vert \mathrm{corona}%
(G)-S\right\vert $.

Lemma \ref{lem2} implies that there is a matching $M$ from $N\left(
\mathrm{core}\left(  G\right)  \right)  $ into $\mathrm{core}\left(  G\right)
$, that can be enlarged to a maximum matching, say $M_{2}$, of $G$.

Since $M_{2}$ matches $\mu\left(  G\right)  $ vertices from $A=\left(
\mathrm{corona}(G)-S\right)  \cup\left(  N\left(  \mathrm{core}\left(
G\right)  \right)  \right)  $ by means of $\mu\left(  G\right)  $\ edges, and
$\left\vert A\right\vert =\mu\left(  G\right)  +1$, it follows that $M_{2}-M$
matches $\left\vert \left(  \mathrm{corona}(G)-S\right)  \right\vert -1$
vertices from $A$ into $S-\mathrm{core}\left(  G\right)  $, because $M$
saturates $N\left(  \mathrm{core}\left(  G\right)  \right)  $ and no edge
joins a vertex of $\mathrm{core}\left(  G\right)  $ to some vertex from
$\mathrm{corona}(G)-S$. Hence, taking into account that $M\cup M_{1}$ is a
matching of $G$, while $M_{2}$ is a maximum matching, we obtain
\begin{gather*}
\mu\left(  G\right)  =\left\vert M_{2}\right\vert =\left\vert N\left(
\mathrm{core}\left(  G\right)  \right)  \right\vert +\left\vert
\mathrm{corona}(G)-S\right\vert -1\leq\\
\leq\left\vert N\left(  \mathrm{core}\left(  G\right)  \right)  \right\vert
+\left\vert S-\mathrm{core}\left(  G\right)  \right\vert =\left\vert
M\right\vert +\left\vert M_{1}\right\vert \leq\mu\left(  G\right)  ,
\end{gather*}
which implies $\left\vert S-\mathrm{core}\left(  G\right)  \right\vert
=\left\vert \mathrm{corona}(G)-S\right\vert -1$.

Finally, we infer that
\begin{gather*}
\left\vert \mathrm{corona}(G)\right\vert =\left\vert S\cup\left(
\mathrm{corona}(G)-S\right)  \right\vert =\alpha\left(  G\right)  +\left\vert
\mathrm{corona}(G)-S\right\vert =\\
=\alpha\left(  G\right)  +\left\vert S-\mathrm{core}\left(  G\right)
\right\vert +1=2\alpha\left(  G\right)  -\left\vert \mathrm{core}\left(
G\right)  \right\vert +1,
\end{gather*}
and this completes the proof.
\end{proof}

\begin{theorem}
If $G$ is a unicyclic non-K\"{o}nig-Egerv\'{a}ry graph, then%
\[
\mathrm{\ker}\left(  G\right)  =\cup\left\{  \mathrm{\ker}\left(
T_{x}\right)  :x\in N_{1}\left(  C\right)  \right\}  =\mathrm{core}(G).
\]

\end{theorem}

\begin{proof}
Since $T_{x}$ is bipartite, by Theorem \ref{th2}\emph{(ii) }implies that
$\ker\left(  T_{x}\right)  =\mathrm{core}(T_{x})$, for every $x\in
N_{1}\left(  C\right)  $.

According to Theorems \ref{th2}\emph{(i)} and \ref{th12}\emph{(iii)}, it
follows that
\[
\mathrm{\ker}\left(  G\right)  \subseteq\mathrm{core}(G)=\cup\left\{
\mathrm{core}\left(  T_{x}\right)  :x\in N_{1}\left(  C\right)  \right\}  .
\]
Hence $A_{x}=\mathrm{\ker}\left(  G\right)  \cap V\left(  T_{x}\right)
\subseteq\mathrm{core}(G)\cap V\left(  T_{x}\right)  =\mathrm{core}\left(
T_{x}\right)  =\ker\left(  T_{x}\right)  $, for every $x\in N_{1}\left(
C\right)  $.

Assume that $A_{q}\neq\ker\left(  T_{q}\right)  $ for some $q\in N_{1}\left(
C\right)  $. It follows that $d\left(  A_{q}\right)  <d\left(  \ker\left(
T_{q}\right)  \right)  $.

Since, by Theorem \ref{th12}\emph{(i)}, we have $\mathrm{\ker}\left(
G\right)  \subseteq\mathrm{core}(G)$, Lemma \ref{lem2}\emph{(ii) }ensures that
$N\left[  V\left(  C\right)  \right]  \cap$ $\mathrm{\ker}\left(  G\right)
=\emptyset$. Consequently, the set $W=\left(  \mathrm{\ker}\left(  G\right)
-A_{q}\right)  \cup\ker\left(  T_{q}\right)  $ is independent, and satisfies%
\begin{gather*}
d\left(  \mathrm{\ker}\left(  G\right)  \right)  =d\left(  \cup\left\{
\mathrm{\ker}\left(  G\right)  \cap V\left(  T_{x}\right)  :x\in N_{1}\left(
C\right)  \right\}  \right)  =\\
=\sum\limits_{x\in N_{1}\left(  C\right)  }d\left(  \mathrm{\ker}\left(
G\right)  \cap V\left(  T_{x}\right)  \right)  =\sum\limits_{x\in N_{1}\left(
C\right)  }d\left(  A_{x}\right)  =d\left(  A_{q}\right)  +\sum\limits_{x\in
N_{1}\left(  C\right)  -\left\{  q\right\}  }d\left(  A_{x}\right)  <\\
<d\left(  \ker\left(  T_{q}\right)  \right)  +\sum\limits_{x\in N_{1}\left(
C\right)  -\left\{  q\right\}  }d\left(  A_{x}\right)  =d\left(  W\right)
\leq\max\left\{  d\left(  X\right)  :X\subseteq V\left(  G\right)  \right\}
=d\left(  \mathrm{\ker}\left(  G\right)  \right)  ,
\end{gather*}
which is a contradiction.

Therefore, we infer that
\[
\mathrm{\ker}\left(  G\right)  \cap V\left(  T_{x}\right)  =\mathrm{core}%
(G)\cap V\left(  T_{x}\right)  =\mathrm{core}\left(  T_{x}\right)
=\ker\left(  T_{x}\right)
\]
hold for each $x\in N_{1}\left(  C\right)  $. Hence,
\[
\mathrm{\ker}\left(  G\right)  =\mathrm{core}(G)=\cup\left\{  \mathrm{core}%
\left(  T_{x}\right)  :x\in N_{1}\left(  C\right)  \right\}  =\cup\left\{
\ker\left(  T_{x}\right)  :x\in N_{1}\left(  C\right)  \right\}  ,
\]
as claimed.
\end{proof}

\begin{remark}
If $G$ is a unicyclic K\"{o}nig-Egerv\'{a}ry graph that is non-bipartite, then
the difference between $\left\vert \mathrm{core}(G)\right\vert $ and
$\left\vert \mathrm{\ker}\left(  G\right)  \right\vert $ may equal any
non-negative integer. For instance, the graph $G_{2k+1}$ from Figure \ref{12}
satisfies $\alpha\left(  G_{2k+1}\right)  =k+3$, $\mu\left(  G_{2k+1}\right)
=k+1$, while $\left\vert \mathrm{core}(G_{2k+1})\right\vert -\left\vert
\mathrm{\ker}\left(  G_{2k+1}\right)  \right\vert =k-1,k\geq1$.
\end{remark}

\begin{figure}[h]
\setlength{\unitlength}{1cm}\begin{picture}(5,1.8)\thicklines
\multiput(4,0.5)(1,0){9}{\circle*{0.29}}
\multiput(10,0.5)(0.2,0){5}{\circle*{0.12}}
\put(5,1.5){\circle*{0.29}}
\put(11,1.5){\circle*{0.29}}
\put(4,0.5){\line(1,0){6}}
\put(5,0.5){\line(0,1){1}}
\put(11,1.5){\line(1,-1){1}}
\put(11,0.5){\line(0,1){1}}
\put(11,0.5){\line(1,0){1}}
\put(2.5,1){\makebox(0,0){$G_{2k+1}$}}
\put(4,0.1){\makebox(0,0){$x$}}
\put(5,0.1){\makebox(0,0){$y$}}
\put(5.3,1.5){\makebox(0,0){$z$}}
\put(6,0.1){\makebox(0,0){$v_{1}$}}
\put(7,0.1){\makebox(0,0){$v_{2}$}}
\put(8,0.1){\makebox(0,0){$v_{3}$}}
\put(9,0.1){\makebox(0,0){$v_{4}$}}
\put(10,0.1){\makebox(0,0){$v_{5}$}}
\put(11,0.1){\makebox(0,0){$v_{2k}$}}
\put(12,0.1){\makebox(0,0){$v_{2k+1}$}}
\end{picture}\caption{$\ker\left(  G_{2k+1}\right)  =\left\{  x,z\right\}  $,
while \textrm{core}$\left(  G_{2k+1}\right)  =\left\{  x,z,v_{1}%
,v_{3},...,v_{2k-1}\right\}  $.}%
\label{12}%
\end{figure}
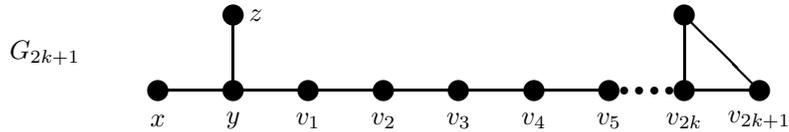

\section{Conclusions}

The equality $\mathrm{core}(G)=\mathrm{\ker}\left(  G\right)  $ may fail for
some non-bipartite unicyclic K\"{o}nig-Egerv\'{a}ry graphs; e.g., the graphs
$G_{1}$ and $G_{2}$ from Figure \ref{fig1122} satisfy $\ker\left(
G_{1}\right)  =\left\{  a,b\right\}  \neq$ \textrm{core}$(G_{1})=\left\{
a,b,c\right\}  $, while $\ker\left(  G_{2}\right)  =$ \textrm{core}%
$(G_{2})=\left\{  x,y,z\right\}  $.

\begin{problem}
Characterize non-bipartite unicyclic K\"{o}nig-Egerv\'{a}ry graphs $G$
satisfying $\mathrm{core}(G)=\mathrm{\ker}\left(  G\right)  $.{}
\end{problem}

The non-unicyclic graphs $G_{1}$ and $G_{2}$ from Figure \ref{fig11} satisfy
$\left\vert \mathrm{corona}(G_{1})\right\vert +\left\vert \mathrm{core}%
(G_{1})\right\vert =2\alpha\left(  G_{1}\right)  $ and $\left\vert
\mathrm{corona}(G_{2})\right\vert +\left\vert \mathrm{core}(G_{2})\right\vert
=2\alpha\left(  G_{2}\right)  +1$. \begin{figure}[h]
\setlength{\unitlength}{1cm}\begin{picture}(5,1.8)\thicklines
\multiput(2,0)(1,0){6}{\circle*{0.29}}
\multiput(2,1)(1,0){2}{\circle*{0.29}}
\multiput(6,1)(1,0){2}{\circle*{0.29}}
\put(2,0){\line(1,0){5}}
\put(2,0){\line(0,1){1}}
\put(2,1){\line(1,0){1}}
\put(3,1){\line(1,-1){1}}
\put(6,0){\line(0,1){1}}
\put(6,1){\line(1,0){1}}
\put(7,0){\line(0,1){1}}
\put(5,0.3){\makebox(0,0){$a$}}
\put(5.7,1){\makebox(0,0){$b$}}
\put(6.7,0.3){\makebox(0,0){$c$}}
\put(1.2,0.5){\makebox(0,0){$G_{1}$}}
\multiput(9,0)(1,0){5}{\circle*{0.29}}
\multiput(9,1)(1,0){2}{\circle*{0.29}}
\multiput(12,1)(1,0){2}{\circle*{0.29}}
\put(9,0){\line(1,0){4}}
\put(9,1){\line(1,0){1}}
\put(9,0){\line(0,1){1}}
\put(10,1){\line(1,-1){1}}
\put(12,1){\line(1,0){1}}
\put(12,0){\line(0,1){1}}
\put(13,0){\line(0,1){1}}
\put(8.2,0.5){\makebox(0,0){$G_{2}$}}
\end{picture}\caption{ \textrm{core}$(G_{1})=\left\{  a,b,c\right\}  $ and
\textrm{core}$(G_{2})=\emptyset$.}%
\label{fig11}%
\end{figure}
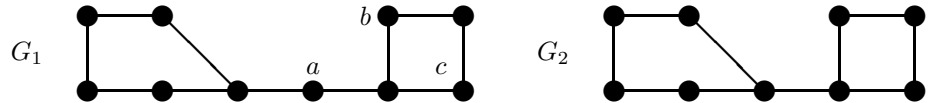

\begin{problem}
Characterize graphs satisfying%
\[
2\alpha\left(  G\right)  \ \leq\left\vert \mathrm{corona}(G)\right\vert
+\left\vert \mathrm{core}(G)\right\vert \leq2\alpha\left(  G\right)  +1.
\]

\end{problem}


\begin{thebibliography}{99}                                                                                               %


\bibitem {Belardo2010}F. Belardo, M. Li, M. Enzo, S. K. Simi\'{c}, J. Wang,
\emph{On the spectral radius of unicyclic graphs with prescribed degree
sequence}, Linear Algebra and its Applications \textbf{432} (2010) 2323-2334.

\bibitem {BorosGolLev}E. Boros, M.C. Golumbic, V. E. Levit, \emph{On the
number of vertices belonging to all maximum stable sets of a graph}, Discrete
Applied Mathematics \textbf{124} (2002) 17-25.

\bibitem {dem}R. W. Deming, \emph{Independence numbers of graphs - an
extension of the K\"{o}nig-Egerv\'{a}ry theorem}, Discrete Mathematics
\textbf{27} (1979) 23-33.

\bibitem {Du2010}Z. Du, B. Zhou, N. Trinajsti\'{c}, \emph{Minimum
sum-connectivity indices of trees and unicyclic graphs of a given matching
number}, Journal of Mathematical Chemistry \textbf{47} (2010) 842-855.

\bibitem {eger}E. Egerv\'{a}ry, \emph{On combinatorial properties of
matrices}, Matematikai Lapok \textbf{38} (1931) 16-28.

\bibitem {Huo2010}B. Huo, S. Ji, X. Li, \emph{Note on unicyclic graphs with
given number of pendent vertices and minimal energy}, Linear Algebra and its
Applications \textbf{433} (2010) 1381-1387.

\bibitem {koen}D. K\"{o}nig, \emph{Graphen und Matrizen}, Matematikai Lapok
\textbf{38} (1931) 116-119.

\bibitem {levm2}V. E. Levit, E. Mandrescu, \emph{Well-covered and
K\"{o}nig-Egerv\'{a}ry graphs}, Congressus Numerantium \textbf{130} (1998) 209-218.

\bibitem {levm3}V. E. Levit, E. Mandrescu, \emph{Combinatorial properties of
the family of maximum stable sets of a graph}, Discrete Applied Mathematics
\textbf{117} (2002) 149-161.

\bibitem {LevMan2003}V. E. Levit, E. Mandrescu, \emph{On }$\alpha^{+}%
$\emph{-stable K\"{o}nig-Egerv\'{a}ry graphs}, Discrete Mathematics
\textbf{263} (2003) 179-190.

\bibitem {LevMan2009a}V. E. Levit, E. Mandrescu, \emph{Greedoids on vertex
sets of unicycle graphs}, Congressus Numerantium \textbf{197} (2009) 183-191.

\bibitem {LevMan5}V. E. Levit, E. Mandrescu, \emph{Critical independent sets
and K\"{o}nig-Egervary graphs},\emph{\ }Graphs and Combinatorics (2011) doi
10.1007/s00373-011-1037-y (in press).

\bibitem {Levman2011a}V. E. Levit, E. Mandrescu, \emph{Vertices belonging to
all critical independent sets of a graph}, (2011) arXiv:1102.0401 [cs.DM] 9
pp. (submitted)

\bibitem {Levman2011b}V. E. Levit, E. Mandrescu, \emph{Critical sets in
bipartite graphs}, (2011) arXiv:1102.1138 [math.CO] 13 pp. (submitted)

\bibitem {LevMan2011c}V. E. Levit, E. Mandrescu, \emph{On the structure of the
minimum critical independent set of a graph}, arXiv:1102.1859 [math.CO] 8 pp. (submitted)

\bibitem {Levman2011d}V. E. Levit, E. Mandrescu, \emph{On the core of a
unicyclic graph}, (2011) arXiv:1102.4727 [cs.DM] 8 pp. (submitted)

\bibitem {Li2010}J. Li, J. Guo, W. C. Shiu, \emph{The smallest values of
algebraic connectivity for unicyclic graphs}, Discrete Applied Mathematics
\textbf{158} (2010) 1633-1643.

\bibitem {ster}F. Sterboul, \emph{A characterization of the graphs in which
the transversal number equals the matching number}, Journal of Combinatorial
Theory Series B \textbf{27} (1979) 228-229.

\bibitem {Wu2010}Y. Wu, J. Shu, \emph{The spread of the unicyclic graphs},
European Journal of Combinatorics \textbf{31} (2010) 411-418.

\bibitem {Zhai2010}M. Zhai, R. Liu, J. Shu, \emph{Minimizing the least
eigenvalue of unicyclic graphs with fixed diameter}, Discrete Mathematics
\textbf{310} (2010) 947-955.

\bibitem {Zhang}C. Q. Zhang, \emph{Finding critical independent sets and
critical vertex subsets are polynomial problems}, SIAM J. Discrete Mathematics
\textbf{3} (1990) 431-438.
\end{thebibliography}
\end{document}